\newif\ifonecol

\ifonecol
\documentclass[peerreview,12pt]{IEEEtran}
\else
\documentclass[twoside, journal]{IEEEtran}
\fi

\usepackage[dvips]{graphicx}
\usepackage{amsmath,amssymb}
\usepackage{graphicx}
\usepackage{placeins}
\usepackage{wrapfig}
\usepackage{verbatim}
\usepackage{array}
\usepackage{flushend}
\usepackage{subcaption}
\usepackage{tikz}
\usepackage{amsthm}
\usepackage{varioref}
\usetikzlibrary{calc}

\usepackage[ruled,vlined,boxed,algonl,linesnumbered]{algorithm2e}

\theoremstyle{theorem}
\newtheorem{theorem}{Theorem}
\newtheorem{corollary}{Corollary}

\newtheorem{proposition}{Proposition}

\theoremstyle{definition}
\newtheorem{definition}{Definition}
\newtheorem{remark}{Remark}
\theoremstyle{remark}
\newtheorem{example}{Example}

\makeatletter
\newcommand*\bigcdot{\mathpalette\bigcdot@{.5}}
\newcommand*\bigcdot@[2]{\mathbin{\vcenter{\hbox{\scalebox{#2}{$\m@th#1\bullet$}}}}}
\makeatother

\SetCommentSty{mycommfont}

\newcommand{\mA}{\mathcal{A}}
\newcommand{\mB}{\mathcal{B}}
\newcommand{\mC}{\mathcal{C}}

\newcommand{\mE}{\mathcal{E}}
\newcommand{\mF}{\mathcal{F}}
\newcommand{\mG}{\mathcal{G}}

\newcommand{\mI}{\mathcal{I}}
\newcommand{\mJ}{\mathcal{J}}

\newcommand{\mS}{\mathcal{S}}
\newcommand{\mT}{\mathcal{T}}

\newcommand{\mW}{\mathcal{W}}
\newcommand{\mX}{\mathcal{X}}
\newcommand{\mY}{\mathcal{Y}}

\newcommand{\bF}{\mathbb{F}}

\newcommand{\bS}{\mathbb{S}}

\newcommand{\bfA}{\mathbf{A}}
\newcommand{\bfB}{\mathbf{B}}

\newcommand{\bfT}{\mathbf{T}}

\newcommand{\tc}{\tilde{c}}
\newcommand{\td}{\widetilde{d}}
\newcommand{\ove}{\overline{\mathcal{E}}}

\newcommand{\tK}{\widetilde{K}}

\newcommand{\ou}{\overline u}
\newcommand{\tP}{\widetilde P}
\newcommand{\tQ}{\widetilde Q}

\newcommand{\tS}{\widetilde{S}}

\newcommand{\uwt}{\underline{\wt}}

\newcommand{\vp}{\varphi}

\DeclareMathOperator{\cs}{cs}
\newcommand{\set}[1]{\left\{{#1}\right\}}

\newcommand{\floor}[1]{\left\lfloor{#1}\right\rfloor}
\newcommand{\bnull}{\mathbf 0}

\newcommand{\hull}[1]{\left\langle{#1}\right\rangle}
\newcommand{\wt}{\mathbf{w}}

\SetKwInOut{KwIn}{Input}
\SetKwInOut{Output}{Output}
\SetKwInOut{KwOut}{Output}
\SetKwInOut{Return}{return}

\graphicspath{{/}}
\begin{document}

\title{Convolutional Polar Kernels}
\author{Ruslan Morozov\\
ITMO University, Saint Petersburg, Russia\\
E-mail: mir4595@yandex.ru}

\sloppy

\maketitle
\begin{abstract}
A family of polarizing kernels is presented together with polynomial-complexity algorithm for computing scaling exponent.
The proposed convolutional polar kernels are based on convolutional polar codes, also known as b-MERA codes.
For these kernels, a polynomial-complexity algorithm is proposed to find weight spectrum of unrecoverable erasure patterns, needed
for computing scaling exponent.
As a result, we obtain scaling exponent and polarization rate for convolutional polar kernels of size up to 1024.
\end{abstract}

\section{Introduction}
Polar codes \cite{arikan2009channel} are the first class of capacity-achieving codes.
They are based on the $N \times N$ Arikan polarizing transformation $A^{(N)}=F^{\otimes M}$,
$N=2^M$, where $F=\begin{pmatrix}1&0\\1&1\end{pmatrix}$ is called the Arikan kernel.
Many other matrices were proposed to replace kernel $F$, together with efficient corresponding kernel processing algorithms
\cite{yao2019explicit,trofimiuk2019reduced}.
Performance of polar codes with given $n\times n$ kernel $K$ depends on properties of matrix $K$,
such as polarization rate and scaling exponent \cite{mondelli2016unified, fazeli2018binary}.

Convolutional polar codes (CvPC, also b-MERA codes) are introduced in \cite{ferris2013branching}.
They are based on convolutional polarizing transformation (CvPT), which is an $n\times n$ matrix, $n=2^m$, which is \textit{not} of the form $K^{\otimes M}$.
They outperform Arikan polar codes under successive cancellation (SC) decoding \cite{prinz2018successive, saber2018convolutional,morozov2018efficient} due to better
polarization properties.

More precisely, consider kernel $K$ and codeword $c_0^{n-1}=u_0^{n-1}K$. On each phase $\vp$, the SC decoder, trying to estimate $u_\vp$, considers probabilities of two cosets: $(\hat u_0^{\vp-1},0,u_{\vp+1}^{n-1})K$ and $(\hat u_0^{\vp-1},1,u_{\vp+1}^{n-1})K$,
where $u_{\vp+1}^{n-1}$ runs over all possible binary vectors of length $n-\vp-1$, and $\hat u_0^{\vp-1}$ are already estimated
symbols. Note that the difference between  (XOR of) any two vectors from the cosets is a vector from the set $C_\vp=\set{(0_0^{\vp-1},1,u_{\vp+1}^{n-1})K}$.
Consider a ``dominating set'' of $C_\vp$, i.e., set  $\overline C_\vp=\set{\overline a_0^{n-1}|\exists a_0^{n-1}\in C_\vp: \forall i: \overline a_i\geq a_i}$.
Note that in the case of BEC, set $\overline C_\vp$ describes all erasure patterns, after which one cannot recover $u_\vp$.
Polarization properties of $K$ depend on the weight distributions of $\overline C_\vp$ for each $\vp$.
In some sense, matrix $Q^{(n)}$ has better weight distributions of $\overline C_\vp$ then the Arikan polarizing transformation
$F^{\otimes m}$ of the same size $n=2^m$.

The weight distributions of $\overline C_\vp$ allow one to obtain scaling exponent and polarization rate of a kernel.
In this paper we derive them for kernel $Q^{(n)}$, based on the recursive expansion $Q^{(n)}=(X^{(n)}Q^{(n/2)},Z^{(n)}Q^{(n/2)})$, where $(A,B)$ means concatenation of matrices $A$ and $B$.
Matrices $X^{(n)}$ and $Z^{(n)}$ are of size $n\times n/2$, and their rank is $n/2$.
They have diagonal-like structure, i.e. all positions of $1$'s are not far from diagonal $\set{(2j,j), 0\leq j<n/2}$, which results in simple recursive relations between weight distributions of $\overline C_\vp$ for  $Q^{(n/2)}$ and $Q^{(n)}$.
In this paper we prove these relations, which lead to an algorithm of computing scaling exponent for $Q^{(n)}$
for any $n$ with polynomial complexity in $n$.


\section{Background}
\label{s:bg}
\subsection{Notations}
The following notations are used in the paper. 
$\bF$ denotes the Galois field of two elements.
For integer $n$ we denote the set $[n]=\{0,1,\ldots n-1\}$.
Symbol $a_b^c$ denotes vector $(a_b,a_{b+1},\ldots, a_c)$.
For $m \times n$ matrix  $A$ and sets $\mX \subseteq [m]$, $\mY \subseteq [n],$ by $A_{\mX,\mY}$ we denote the submatrix of $A$ with rows from set $\mX$ and columns from set $\mY$, where indexing of rows and columns starts from zero.
Notation $c_{\mX}$ is defined similarly for vector $c$.
If  $\mX=*$ or $\mY=*$, this means that all rows or all columns of the original matrix are in the submatrix. 
Symbol $A_{\overline \mX,\overline \mY}$ denotes a submatrix of $A$ consisting  of rows and columns with indices that are not in $\mX$ and $\mY$, respectively.
The vector of $i$ zeroes is denoted by $\bnull^i$, or just by $\bnull$, if $i$ is clear from the context.
We also use symbol $(a,b)$ for concatenation of vectors/matrices/elements $a$ and $b$.
Also we use strings of $0$'s and $1$'s for an explicit binary vector, e.g. $110=(1,1,0)$.


\subsection{Polar Codes }
\label{ss:sc}
In this paper we consider polar codes, defined as a set of vectors
\begin{align}
c_{[N]}=u_{[N]}K^{\otimes M}, u_{\mF}=\bnull^{N-k}, u_{\mI}\in\bF^k,
\label{eq:pcdef}
\end{align}
where $K$ is an $n\times n$ invertible matrix over $\bF$, which is not upper-triangular under any column permutation, $\mF\subset [N]$, $|\mF|=N-k$, $\mI=[N]\setminus\mF$, and symbol $K^{\otimes M}$ denotes the $M$-times Kronecker product of $K$ with itself.
The length of the code is $N=n^M$, the dimension is $k$.
Matrix $K$ is called the kernel.

Consider transmission of codeword $c_{[N]}=\ou_{[N]}K^{\otimes M}$ through a binary-input memoryless channel $\mW:\bF\to\mY$.
The SC\ decoding algorithm makes successive estimations $\hat u_\vp$ of symbols $\overline u_\vp$, $\vp\in[N]$.
On phase $\vp$, for $u_\vp\in\bF$ the SC decoding algorithm calculates the value of $W^{(\vp)}_N(y_0^{N-1}, \hat u_{[\vp]}|u_\vp)$, defined as
\begin{align}
W^{(\vp)}_N(y_{[N]}, u_{[\vp]}|u_{\vp}) = 2^{-N}\cdot\!\!\!\!\!\!\!\!\!\!\sum_{u_{\vp+1}^{N-1} \in \bF^{N-\vp-1}}\!\!\!\!\!\!\!\!\! \mW^N(y_{[N]}|u_{[N]}K^{\otimes M}),
\label{eq:wdef}
\end{align}
where $\mW^N(y_{[N]}|c_{[N]})=\prod_{i=0}^{N-1}\mW(y_i|c_i)$.
Then, the estimation of $\overline u_\vp$ is made by
\begin{align}
\hat u_\vp=\begin{cases}
0, &\vp\in\mF \\
\arg \displaystyle\max_{u_{\vp}\in\bF}W^{(\vp)}_N(y_{[N]},\hat u_{[\vp]}|u_\vp) , &\vp \in \mI.
\end{cases}
\label{eq:hd}
\end{align}

Computing \eqref{eq:wdef} can be done recursively by
\ifonecol
\else
\begin{align}
&W^{(ni+j)}_N(u_0^{ni+j}|y_0^{N-1})=\nonumber\\
&\sum_{u_{ni+j+1}^{ni+n-1}}\prod_{s=0}^{n-1}\!W_{N/n}^{(j)}\left((u_{nt}^{nt+n-1}K)_s,t\in[j\!+\!1]\big|y_{N/ns}^{N/ns+N/n-1}\right)\!.
\label{eq:wrec}
\end{align}
\fi
If transmitted $\ou_i\in\bF$ are uniformly distributed, then \eqref{eq:wrec} is equal to \eqref{eq:wdef} multiplied by a constant which does not affect maximization \eqref{eq:hd}.
Computing \eqref{eq:wrec} on one layer of recursion for all $j\in[n]$ is called \textit{kernel processing}.
\subsection{Scaling Exponent and Polarization Rate}
In this paper we consider two polarization properties of a kernel, namely, scaling exponent and polarization rate, which can be used to estimate performance of polar codes with a given kernel.

Polar codes are based on the polarization phenomenon, i.e., some part of channels $W^{(\vp)}_N$ tend to the noiseless channel, and others tend to complete noise with $N\to\infty$.
The Bhattacharyya parameter of a binary-input channel $W$ with output alphabet $\mY$ is used as an upper bound on error probability of channel $W$.
It is defined as
\begin{align}
Z(W)=\sum_{y\in\mY}\sqrt{W(y|0)W(y|1)}.
\label{eq:bhadef}
\end{align}

\textit{Scaling exponent} \cite{fazeli2014scaling,hassani2014finitelength} is defined for channel $W$ and kernel $K$ as number $\mu(W,K)$, such that there exists a finite non-zero value of
\begin{align}
\lim_{N\to\infty}\frac{\#\set{i|\epsilon < Z(W^{(i)}_N)<1-\epsilon'}}{N}\cdot N^{1/\mu(W,K)}
\end{align}
for any $0<\epsilon<1-\epsilon'<1$, where $N=n^M$.
Such number is not yet proven to exist.
We assume it exists (this assumption is also known as the scaling assumption \cite{yao2019explicit}).

\textit{Polarization rate} is defined for a kernel (independent of the underlying channel) as number $E(K)$, such that:
\begin{align*}
\forall \beta < E(K):& \liminf_{N\to\infty}\frac{\#\set{i|Z(W^{(i)}_N)\leq 2^{-n^{N\beta}}}}{N}=I(W),\\
\forall \beta > E(K):& \liminf_{N\to\infty}\frac{\#\set{i|Z(W^{(i)}_N)\geq 2^{-n^{N\beta}}}}{N}=1,
\end{align*}
where $I(W)$ denotes the capacity of channel $W$.

\subsection{Convolutional Polarizing Transformation}
Convolutional polar codes \cite{ferris2017convolutional} (CvPCs) are a family of linear block codes of length $n=2^m$.
The generator matrix of a CvPC consists of rows of $n \times n$ non-singular matrix $Q^{(n)}$, called convolutional polarizing transformation
(CvPT), defined as
\begin{align}
Q^{(n)}=\left(X^{(n)}Q^{(n/2)},Z^{(n)}Q^{(n/2)}\right),
\label{eq:qdef}
\end{align}
where $Q^{(1)}=(1)$, $X^{(l)}$ and $Z^{(l)}$ are $l\times l/2$ matrices, defined for even $l$ as
\begin{align}
&X^{(l)}_{i,j} =\begin{cases} 1, & \text{if } 2j\leq i \leq 2j+2\\
0, & \text{otherwise}
\end{cases}
\label{eq:xdef}
\\
&Z^{(l)}_{i,j} =\begin{cases} 1, & \text{if } 2j< i \leq 2j+2\\
0, & \text{otherwise}
\end{cases}
\label{eq:zdef}
\end{align}
For example, 
$$X^{(8)}=\begin{pmatrix}11100000\\00111000\\00001110\\00000011\end{pmatrix}^T,
Z^{(8)}=\begin{pmatrix}01100000\\00011000\\00000110\\00000001\end{pmatrix}^T.$$
Expansion \eqref{eq:qdef} corresponds to one layer of the CvPT, which is depicted in Fig.~\ref{fig:cvpt}.
The $m$-th layer of the CvPT is a mapping of vector $u_0^{n-1}$ to vectors $x_0^{n/2-1}=u_0^{n-1}X^{(n)}$ and $z_0^{n/2-1}=u_0^{n-1}Z^{(n)}$,
where
\ifonecol
\else
\begin{align}
x_i=u_{2i}+u_{2i+1}+u_{2i+2},\; &z_i=u_{2i+1}+u_{2i+2},i\leq\frac{n}{2}-2; \nonumber\\
x_{n/2-1}=u_{n-2}+u_{n-1}, \;&z_{n/2-1}=u_{n-1}.
\label{eq:xz}
\end{align}
\fi
\begin{figure}
\centering
\ifonecol
\begin{tikzpicture}[x=.8cm,y=.8cm]
\else
\begin{tikzpicture}[x=0.475cm,y=0.35cm]
\small
\fi

\def\r{0.2}
\def\g{0.43}


\foreach \x in {1,...,4,6,7,...,9}{
  \draw (1,\x)--(4,\x);
}
\foreach \x in {1,3,4,5,8,9}{
\draw (6,\x)--(8,\x);
}
\foreach \x in {1,3,4,5,8,9}{
\draw (11,\x)--(11.5,\x);
}


\foreach \x in {2,7}{
\draw (1.5,\x)--(1.5,\x+1+\r);
\draw (1.5,\x+1) circle (\r);
}

\foreach \x in {1,3,8,6}{
\draw (2,\x)--(2,\x+1+\r);
\draw (2,\x+1) circle (\r);
}


\draw (4,0) rectangle (6,10);
\draw (8,0.75) rectangle (11,4.25);
\draw (8,4.75) rectangle (11,9.25);


\node at (2.5,5) {\ldots};
\node at (7,6.5) {\ldots};
\node at (7,2) {\ldots};
\node at (12,6.5) {\ldots};
\node at (12,2) {\ldots};

\node[rotate=45] at (5,5) {permute};

\foreach \x in {0,...,3}{
  \node at (0.2,9-\x) {$u_\x$};
}
\foreach \x in {0,1}{
\node at (3.125,9+\g-2*\x) {$x_\x$};
\node at (3.125,8+\g-2*\x) {$z_\x$};
}
\foreach \x in {1,...,4}{
  \node at (0.2,\x) {$u_{n-\x}$};
}
\foreach \x in {1,2}{
\node at (3.125,-1+\g+2*\x) {$z_{n-\x}$};
\node at (3.125,\g+2*\x) {$x_{n-\x}$};
}

\node at (7,1+\g) {$z_{\frac{n}{2}-1}$};
\node at (7,3+\g) {$z_{1}$};
\node at (7,4+\g) {$z_{0}$};
\node at (7,5+\g) {$x_{\frac{n}{2}-1}$};
\node at (7,8+\g) {$x_1$};
\node at (7,9+\g) {$x_0$};

\node at (9.5,7) {$Q^{(n/2)}$};
\node at (9.5,2.5) {$Q^{(n/2)}$};

\node at (12.5,0.9) {$c_{n-1}$};
\node at (12.5,2.9) {$c_{n/2+1}$};
\node at (12.5,3.9) {$c_{n/2}$};
\node at (12.5,4.9) {$c_{n/2-1}$};
\node at (12.5,7.9) {$c_1$};
\node at (12.5,8.9) {$c_0$};

\end{tikzpicture}
\caption{Convolutional polarizing transformation $Q^{(n)}$}
\label{fig:cvpt}
\end{figure}
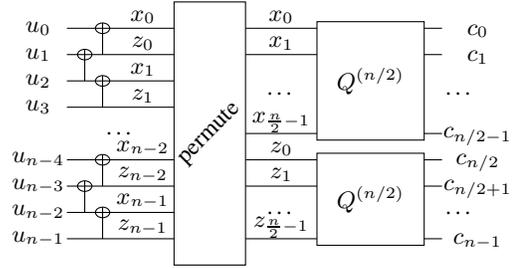

\subsection{Polarization Behavior (PB)}
For a given kernel $K$, scaling exponent for BEC\ and polarization rate 
can be obtained from so-called polarization behaviour, which is defined as follows.

Consider transmission of codeword $c_0^{n-1}=u_0^{n-1}K$ through BEC $\mW$.
Denote by $\mE\subseteq [n]$ the erasure configuration, i.e., the set of erased positions of $c_0^{n-1}$.
Consider phase $\vp$ of SC decoding. Assume that all $u_0^{\vp-1}$ was estimated correctly.
Assume for simplicity $u_0^{\vp-1}=\bnull^\vp$ (otherwise we can set $\tilde c_0^{n-1}=c_0^{n-1}+u_0^{\vp-1}K_{[\vp],*}$).
Each non-erased symbol $c_j, j\in\overline\mE=[n]\setminus\mE$  can be expressed as 
$
c_j=\sum_{i=\vp}^{n-1}u_iK_{i,j}=u_i^{n-1}K_{\overline{[\vp]},\set{j}}, \;j\in\overline\mE,
$
where symbol $\bigcdot$ denotes dot product of two vectors with the same dimension over $\bF$. Given $c_{\ove}$, the receiver can compute any linear combination $\sum_{j\in\ove} b_jc_j$, which is also a linear combination of input symbols $u_{\vp}^{n-1}$. The receiver can recover any linear combination of the form
\ifonecol
\else
\begin{align}
&\sum_{j\in\ove}b_jc_j=\sum_{j\in\ove}b_j\sum_{i=\vp}^{n-1}u_iK_{i,j}=\sum_{i=\vp}^{n-1}u_i\sum_{j\in\ove}b_jK_{i,j}
\nonumber\\&
=u_{\vp}^{n-1}\ \bigcdot p_0^{n-\vp-1},\; p_0^{n-\vp-1}\in\cs\hat K,
\label{eq:rec}
\end{align}
\fi
where $\hat K=K_{\overline{[\vp]},\ove}$ and $\cs \hat K$ denotes the column space of matrix $\hat K$.
Symbol $u_{\vp}$ corresponds to linear combination $u_{\vp}^{n-1}\bigcdot~(1,\bnull^{n-\vp-1})$.
Thus, $u_{\vp}$ is erased iff $(1,0,...,0)\notin \cs \hat K$.

\begin{definition}
\label{d:pb}
\textit{Polarization behavior (PB)} of $n\times n$ kernel $K$ is a collection of $n$ polynomials $P^{(0)}(x),...,P^{(n-1)}(x)$, where each polynomial  $P^{(\vp)}(x)=\sum_{w=0}^{n}A_wx^w$ is the weight enumerator of erasure configurations that erase $u_\vp$:
$$
A_w=\left|\set{\mE\subseteq[n]\ \big|\ (1,\bnull^{n-\vp-1})\notin\cs K_{\overline{[\vp]},\ove} \text{ and }|\mE|=w}\right|.
$$
\end{definition}

Knowing PB, one can compute scaling exponent for BEC by the algorithm presented in \cite{hassani2014finitelength}.
In the following section, we present an algorithm for computing PB of $K=Q^{(n)}$.


\section{Computing Scaling Exponent for Convolutional Polar Kernel \label{s:se}}
\subsection{General Description of the Algorithm}
Our algorithm for computing scaling exponent for CvPK consists of three steps:
\begin{enumerate}
\item Compute generalized polarization behaviour (GPB) of CvPK by the recursion, described in Section~\ref{s:gpbcvpk}.
\item Convert GPB to PB, as given in Section~\ref{s:gpb2pb}.
\item Given PB for CvPK, compute scaling exponent for BEC by the algorithm, presented in \cite{hassani2014finitelength} (we do not describe it in this paper).
\end{enumerate}

The proposed algorithm is similar to the algorithm in \cite{morozov2019distance} for computing partial distances of CvPT.
After publishing \cite{morozov2019distance} we found that partial distances of CvPT can be computed with much simpler algorithm \cite{morozov2019simplified}.
However, computing PB of CvPK requires one to fully employ the idea of \cite{morozov2019distance}.
Furthermore, we believe that our approach can be extended to compute PB for an arbitrary kernel.

We provide a list of variables, used in this section, in Table~\ref{t:not} to simplify the reader's life.
\begin{table}
\centering
\caption{The summary of notations.}
\label{t:not}
\begin{tabular}{|m{0.21\linewidth}|m{0.64\linewidth}|}
\hline
$\bF$ & The binary field
\rule{0pt}{8pt}
\\\hline
kernel $K$ & Any non-singular binary $n\times n$ matrix which is not upper-triangular under any column permutation
\\\hline
$\cs A$ & The column space of matrix $A$
\rule{0pt}{9pt}
\\\hline
$[n]$ & Set $\set{0,1,...,n-1}$
\rule{0pt}{9pt}
\\\hline
$\overline \mS$ & For a set $\mS\subseteq[n]$, the complement to $[n]$
\rule{0pt}{9pt}
\\\hline
$a_{\mA}$ & \rule{0pt}{8pt}A subvector of vector $a_0^{t-1}=a_{[t]}$ with ascending indices from set $\mA\subseteq[t]$%
\\\hline
$u_{[n]}$ & Input vector, which is multiplied by kernel $K$
\\\hline
$c_{[n]}$ & Output vector $c_{[n]}=u_{[n]}K$
\rule{0pt}{9pt}
\\\hline
$\vp$ & The phase of SC decoding; the number of first elements of $u$ that we have already estimated correctly.
Due to linearity we assume $u_{[\vp]}=\bnull$
\\ \hline
erasure configuration $\mE$ & \rule{0pt}{7pt}%
The set of erased positions $\mE\subseteq [n]$ of $c_{[n]}$.
After erasures, the receiver knows $c_{\ove}$
\rule{0pt}{7pt}
\\\hline
$\mE'$, $\mE''$ &
\rule{0pt}{7pt}%
Given the erasure configuration $\mE$ of $c_{[n]}$, $\mE'$ is the e.c. of $c_{[n/2]}$ and $\mE''$ is the e.c. of $c_{n/2}^{n-1}$
\rule{0pt}{5pt}
\\\hline
$P^{(\vp)}(x)$ & For an $n\times n$ kernel $K$, the weight enumerator polynomial of erasure configurations of $c_{[n]}=u_{[n]}K$
that erase input symbol $u_\vp$.
Monomial $ax^b$ means that there are $a$ such erasure configurations of cardinality $b$
\\\hline
PB, polarization behaviour (Def.~\ref{d:pb}) & The collection of $P^{(\vp)}(x)$ for each $\vp$
\\\hline
$\bS_J$ & The set of all linear subspaces of $\bF^J$ ($\bS_J\subseteq 2^{\bF^J}$)%
\rule{0pt}{10pt}
\\\hline
$a\bigcdot b$ & Dot product $\sum_i a_ib_i$ of vectors $a$ and $b$%
\rule{0pt}{8pt}
\\\hline
$(\mE,\vp)$-recoverable vector (Def.~\ref{d:chi})& \rule{0pt}{8pt}%
Any vector $p\in\bF^3$, s. t. the value of $p\bigcdot u_{\vp}^{\vp+2}$ can
be computed from subvector $c_{\ove}$ of codeword $c_{[n]}=u_{[n]}K$. This condition is equivalent to $(p,\bnull)\in\cs K_{\overline{[\vp]}, \ove}$ 
\\\hline
$\chi_{\vp}(\mE)$ (Def.~\ref{d:chi}) &
\rule{0pt}{8pt}%
The set of all $(\mE,\vp)$-recoverable vectors (the kernel is assumed to be clear from the context)
\\\hline
$P^{(\vp,\mS)}(x)$ & For an $n\times n$ kernel $K$, the weight enumerator polynomial of erasure configurations $\mE$ for
which $\chi_\vp(\mE)=\mS$
\\\hline
GPB, generalized PB (Def.~\ref{d:gpb}) &
The collection of polynomials $P^{(\vp,\mS)}(x)$ for each $\vp\in [n-2]$ and $\mS\in\bS_3$.
\\\hline
$\hull{001,101}$ & The set of all linear combinations of vectors listed inside $\hull{}$. By default, $\hull{}=\set{\bnull}$
\\\hline
\end{tabular}
\end{table}

\subsection{Generalized Polarization Behaviour (GPB)}
Polarization behaviour (PB)\ characterizes weight spectrum of erasure configurations that erase $u_\vp$.
We found no simple recursion for convolutional polar kernel $K=Q^{(n)}$, that, given PB of $Q^{(n/2)}$, allows one to obtain PB of $Q^{(n)}$.
However, we can obtain recursive formulae for enumerators which count erasure configurations that erase some linear combinations of symbols $u_{\vp}^{\vp+2}$.
Thus, after we generalize the definition of PB to GPB, the GPB of $Q^{(n)}$ can be computed recursively and then converted
to PB.

Assume that the receiver knows $u_0^{\vp-1}$.
Consider linear combination $p_0^2\bigcdot u_{\vp}^{\vp+2}$ of three adjacent input symbols $u_{\vp}^{\vp+2}$ for some given $p_0^2\in\bF^3$.
Recalling \eqref{eq:rec}, one can see that this linear combination can be recovered after erasure configuration $\mE$ iff $(p_0^2,\bnull^{n-\vp-3})\in \cs\hat K$, where $\hat K=K_{\overline{[\vp]},\ove}$.
\begin{definition}
\label{d:chi}
Vector $p_0^{2}$ is $(\mE,\vp)$\textit{-recoverable vector} for kernel $K$, iff $(p_0^2,\bnull^{n-\vp-3})\in\cs K_{\overline{[\vp]}, \ove}$.
The set of  $(\mE,\vp)$-recoverable vectors is denoted by $\chi_{\vp}(\mE)$ (following Greek word $\chi\acute{\omega}\rho o\varsigma$
meaning ``space'').
\end{definition}

It is easy to see that the set $\chi_\vp(\mE)$ is indeed a linear subspace of $\bF^3$, which we write as $\chi_\vp(\mE)\in\bS_3$,
denoting by $\bS_3$ the set of all linear subspaces of $\bF^3$. Throughout the paper, a  subspace of $\bF^3$ is specified by its basis vectors, which are comma-separated strings of $0$ and $1$ listed inside triangular brackets, e.g. $\hull{001, 110}=\hull{001, 111}=\set{\bnull^3, 001, 110, 111}$.

For the sake of convenience, attach index $i\in[16]$ to each subspace $\mT_i\in\bS_3$ of $\bF^3$ (see Table~\ref{t:gpb4}).

In the case of $Q^{(4)}$, $c_{[4]}=(u_0+u_1+u_3,u_0+u_2,u_1+u_2,u_0+u_1+u_2+u_3)$.
After each erasure configuration $\mE\subseteq[4]$ the receiver knows $c_j$ for all $j\notin \mE$, and it can compute all
linear combinations (LCs) of symbols $c_j$.
These LCs correspond to some linear combinations of $u_{[4]}$. 

On phase $\vp=0$, we are interested only in LCs of $u_0^2$, i.e., expressions $p_{[4]}\bigcdot u_{[4]}$ which do not include $u_3$, or, equivalently, when $p_3=0$.
All such $p$'s constitute some set $\mT_i=\chi_0(\mE)$.

For the case of $\vp=1$, we assume that we know exactly the value of $u_0$ and we can subtract it from $c_0$.
Thus, we can assume that $u_0=0$ and $\tc=(u_1+u_3,u_0+u_2,u_1+u_2,u_0+u_1+u_2+u_3)$.
After erasure configuration $\mE$, the receiver knows $\tc_j$, $j\notin\mE$, and all their linear combinations, which lead
to $p_{[3]}\bigcdot u_1^3$ for some $p$'s.
All such $p$'s form the set $\mT_i=\chi_1(\mE)$.

\begin{example}
Let us compute $\chi_0(\set{0,3})$ for $Q^{(4)}$.
In this case, $\vp=0$, $\mE=\set{0,3}$, $c_{\ove}=c_{\set{1,2}}$ and
\ifonecol
\else
\begin{align*}
&K=Q^{(4)}=\begin{pmatrix} 1000\\1010\\0110\\1111\end{pmatrix},\hat K=Q^{(4)}_{*, \set{1,2}}=\begin{pmatrix}00\\01\\11\\11\end{pmatrix},\\
&c_{\set{1,2}}=(u_2+u_3, u_1+u_2+u_3).
\end{align*}
\fi
After erasures, the receiver knows $u_2+u_3$ and $u_1+u_2+u_3$, which are not linear combinations of symbols $u_0^2$ as they
include $u_3$. However, the sum $c_1+c_2=u_1$ is a linear combination $p_0^2 \bigcdot
u_0^2$ with $p_0^2=(010$).
Thus, $\chi_0(\set{0,3})=\hull{010}$.
Another way of thinking is to observe that $\cs \hat K=\set{\bnull^4, 0011, 0111, 0100}$.
Vectors, corresponding to linear combinations of $u_0^2$, have the last zero element.
These vectors are $\set{\bnull^4, 0100}$.
Removing the last element, which corresponds to the zero coefficient before $u_3$,  we obtain $\chi_0(\set{0,3})=\set{\bnull^3, 010}=\hull{010}$.
\end{example}

Consider also the mapping $\chi^{-1}:\bS_3\to 2^{2^{[n]}}$, the inverse image of $\chi$.
In words, $\chi^{-1}_\vp(\mS)$ is the set of all erasure configurations, after which the receiver can recover linear combination $p_0^2 \bigcdot u_\vp^{\vp+2}$ if \textit{and only if} $p_0^2\in\mS$.

We can imagine this mapping as dividing all $\mE\subseteq[n]$ into $16$ ``boxes'', the $i$-th box contains those $\mE$ for which $\chi_\vp(\mE)=\mT_i$.
Thus, the $i$-th box contains exactly $\chi^{(-1)}_\vp(\mT_i)$.
\begin{example}
Let us compute $\chi_0^{-1}(\hull{110})$ for $Q^{(4)}$.
In this case, $\vp=0$, $\mS=\set{\bnull^3,110}$.
The set $\chi^{-1}_0(\hull{110})$ is the set of erasure configurations, after which the receiver can recover $u_0+u_1$ (and no other non-zero linear combination of $u_0^2$).
Consider erasure configuration $\mE_0=\set{2}$.
The receiver knows $(c_0,c_1,c_3)=(u_0+u_1+u_3,u_2+u_3,u_3)$.
It can recover $u_0+u_1=c_0+c_3$.
But it can also recover $u_2=c_1+c_3$ and $u_0+u_1+u_2=c_0+c_1$ and others, so the space corresponding to $\mE_0$ is not $\mS$, though it contains it
as a proper subset.
If we erase positions $\mE_1=\set{1,2}$, the receiver knows $(c_0,c_3)=(u_0+u_1+u_3,u_3)$, and it can compute only $c_0+c_3=u_0+u_1$.
It can be seen that there is no other erasure configuration, which leads to knowing $u_0+u_1$ and erasing all other linear combinations of symbols $u_0^2$.
So, $\chi^{-1}_0(\hull{110})=\set{\set{1,2}}$.
\end{example}

\begin{definition}
\label{d:gpb}
A \textit{generalized polarization behaviour} (GPB) for kernel $K$ is a collection of polynomials $P^{(\vp, \mS)}(x)=\sum_{w=0}^n P^{(\vp,\mS)}_wx^w$ for each
$\vp\in[n-2]$ and each $\mS\in\bS_3$, such that
\begin{align}
P^{(\vp,\mS)}_w=\left|\set{\mE\subseteq[n]\;\big| \; \chi^{(n)}_\vp(\mE)=\mS\text{ and } |\mE|=w}\right|.
\label{eq:gpbdef}
\end{align}
\end{definition}
In other words, $P^{(\vp,\mS)}(x)$ is the weight enumerator polynomial of erasure configurations in $\chi_\vp^{-1}(\mS)$.
\begin{table}
\caption{The GPB of $Q^{(4)}$}
\label{t:gpb4}
\centering
\begin{tabular}{|c|c|c|c|c|c|c|c|}
\hline
$i$&$\mT_i$&$\!\!P^{(0,\mT_i)}\!\!$&$\!\!\!P^{(1,\mT_i)}\!\!\!$&$i$&$\mT_i$&$\!\!\!P^{(0,\mT_i)}\!\!\!$&$\!\!\!P^{(1,\mT_i)}\!\!\!$
\\\hline
0&$\set{\bnull}$&$\!\!\!x^4+4x^3\!\!\!$&$x^4$ & 8&$\hull{100,010}$&$0$&$0$    
\\\hline                                             
1&$\hull{100}$&$0$&$0$            & 9&$\hull{100,001}$&$0$&$x^2$  
\\\hline                                             
2&$\hull{010}$&$x^2$&$0$          & 10&$\hull{010,001}$&$x$&$x^2$ 
\\\hline                                            
3&$\hull{001}$&$x^2$&$x^3$        & 11&$\hull{110,001}$&$x$&$x^2$ 
\\\hline                                              
4&$\hull{110}$&$x^2$&$0$          & 12&$\hull{100,011}$&$0$&$x^2$ 
\\\hline                                            
5&$\hull{101}$&$x^2$&$x^3$        & 13&$\hull{101,010}$&$x$&$x^2$ 
\\\hline                                             
6&$\hull{011}$&$x^2$&$x^3$        & 14&$\hull{110,101}$&$x$&$x^2$ 
\\\hline                                           
7&$\hull{111}$&$x^2$&$x^3$        & 15&$\bF^3$&$1$&$\!\!\!4x+1\!\!\!$         
\\\hline                                            
\end{tabular}
\end{table}

\begin{example}
\label{e:gpb}
The GPB of $Q^{(4)}$ is given in Table~\ref{t:gpb4}.
The GPB consists of polynomials $P^{(\vp,\mT_i)}(x)$ for $\vp\in[2]$ and $i\in[16]$.
\end{example}

\subsection{Recursive Computation of GPB \label{s:gpbcvpk}}

Assume that we know GPB for kernel $Q^{(n/2)}$.
Recall that $c_0^{n-1}=u_0^{n-1}Q^{(n)}=(x_0^{n/2-1}Q^{(n/2)},z_0^{n/2-1}Q^{(n/2)})$.
Consider linear combination $p_0^2 \bigcdot u_{\vp}^{\vp+2}$ for some $p_0^2\in\bF^3$.
Denote the erasure configurations of left and right half of $c_0^{n-1}$ by $\mE'=\mE\cap [n/2]$ and $\mE''=\set{j-\frac{n}{2}|j\in\mE, j\geq \frac{n}{2}}$.
Then, all recoverable $p \bigcdot u_{\vp}^{\vp+2}$ follow from recoverability of $p'\bigcdot x_{\psi}^{\psi+2}$ and $p''\bigcdot z_{\psi}^{\psi+2}$ for erasure configurations $\mE'$ and $\mE''$, respectively, for some particular $\psi\in[\frac{n}{2}]$, $p',p''\in\bF^3$.
This connection is given by the following theorem.
\begin{theorem}
\label{t1}
Consider kernel $Q^{(n)}$, defined in \eqref{eq:qdef}--\eqref{eq:zdef}, $n\geq 8$.
For given $\mE\subseteq[n]$ and $0\leq \vp \leq n-3$, vector $p_0^2$ is $(\mE,\vp)$-recoverable iff 
\begin{align}
\exists p',p''\in\bF^3: (p_0^2,\bnull^{J_\vp})=p'A_{\vp}+p''B_\vp,
\label{eq:t1}
\end{align}
where $p'$ and $p''$ are $(\mE',\psi)$-recoverable and $(\mE'',\psi)$-recoverable for kernel $Q^{(n/2)}$ and $\psi=\max\set{0,\floor{\frac{\vp-1}{2}}}$.
The values of $J_\vp$, $A_\vp$, $B_\vp$ depend on $\vp$ as follows. For $\psi\leq\frac{n}{2}-3$:
\allowdisplaybreaks
\begin{align}
\label{eq:t1vp0}
&J_0\!=\!3, A_0=\bfA=\begin{pmatrix}111000\\001110\\000011\end{pmatrix},B_0=\bfB=\begin{pmatrix}011000\\000110\\000001\end{pmatrix}\\
\label{eq:t1vp1}
&J_{2\psi+1}=2, A_{2\psi+1}=\bfA_{*,\overline{[1]}}, B_{2\psi+1}=\bfB_{*,\overline{[1]}}\\
\label{eq:t1vp2}
&J_{2\psi+2}\!=\!1, A_{2\psi+2}=\bfA_{*,\overline{[2]}}, B_{2\psi+2}=\bfB_{*,\overline{[2]}}\\
\label{eq:t1vp3}
&J_{n-3}\!=\!0, A_{n-3}=\bfA_{*,\overline{[3]}},\normalsize B_{n-3}=\bfB_{*,\overline{[3]}}
\end{align}
\end{theorem}
\begin{proof}
The proof is in the Appendix~\ref{a:t1}.
\end{proof}

Theorem~\ref{t1} defines the relation between subspaces of known linear combinations of symbols $x_{\psi}^{\psi+2}$ and $z_{\psi}^{\psi+2}$ and subspace of known linear
combinations of $u_{\vp}^{\vp+2}$ for some given erasure configuration $\mE$.
Applying this relation to each $\mE\subseteq[n]$, one can compute weight enumerators of erasure configurations for each possible subspace of
linear combinations of symbols $u_{\vp}^{\vp+2}$ by the following theorem.

\begin{theorem}
\label{t2}
For given $n\geq8$, $\vp\in[n-2]$, consider the transformation $\bfT_\vp:\bS_3\times\bS_3\to\bS_3$, which maps spaces of
$p'$ and $p''$ to space of all possible $p$'s defined by \eqref{eq:t1}:
\ifonecol
\else
\begin{align}
&\bfT_\vp(\mS',\mS'')=\nonumber\\
&\set{p_0^2\big|\exists p'\in\mS', p''\in\mS'': (p_0^2,\bnull^{J_\vp})=p'A_\vp+p''B_\vp},
\label{eq:tdef}
\end{align}
\fi
where $J_\vp$, $A_\vp$, $B_\vp$ are given  in \eqref{eq:t1vp0}--\eqref{eq:t1vp3}.
Denote by $P^{(\vp, \mS)}(x)$ the GPB of kernel $Q^{(n)}$ for phase $\vp$, and by $R^{(\psi, \mS)}(x)$ the GPB of kernel $Q^{(n/2)}$ for phase $\psi=\max\set{0,\frac{\vp-1}{2}}$.
Then,
\begin{align}
P^{(\vp, \mS)}(x)=\sum_{(\mS',\mS'')\in\bfT_\vp^{-1}(\mS)}R^{(\psi,\mS')}(x)\cdot R^{(\psi,\mS'')}(x),
\label{eq:t2}
\end{align}
where $\bfT_\vp^{-1}:\bS_3\to 2^{\bS_3\times \bS_3}$ is the inverse image of $\bfT_\vp$.
\end{theorem}
\begin{proof}
The proof is in the Appendix~\ref{a:t2}.
\end{proof}
\begin{example}
On one hand, one can straightforwardly compute $\bfT_{\vp}(\hull{010},\hull{110,001})$ for the case of odd $\vp=2\psi+1$.
Values of $\mS'=\hull{010}$ and $\mS''=\hull{110,001}$ mean that, given values of $(x_0^{n/2-1}X^{(n)})_{\overline{\mE'}}$ and $(z_0^{n/2-1}Z^{(n)})_{\overline{\mE''}}$, the receiver knows
\begin{align*}
&f_0=(010)\bigcdot x_{\psi}^{\psi+2}=x_{\psi+1}=u_{\vp+1}+u_{\vp+2}+u_{\vp+3}, \\
&f_1=(110)\bigcdot z_{\psi}^{\psi+2}=z_{\psi}\!+\!z_{\psi+1}=u_{\vp}\!+\!u_{\vp+1}\!+\!u_{\vp+2}\!+\!u_{\vp+3},\\
&f_2=z_{\psi+2}=u_{\vp+3}+u_{\vp+4}.
\end{align*}
Now we must find linear combinations of symbols $f_0^2$, which involve only  symbols $u_{\vp}^{\vp+2}$.
There is only one such non-zero linear combination: $f_0+f_1=u_\vp=(100)\bigcdot u_\vp^{\vp+2}$.
This means that $\bfT_{2\psi+1}(\hull{010},\hull{110,001})=\hull{100}$.

On the other hand, we can compute the same value via Theorem~\ref{t2}:
\ifonecol
\else
\begin{align*}
&\set{p'A_{2\psi+1}+p''B_{2\psi+1}}_{p'\in\mS',p''\in\mS''}=\\
&\set{\bnull^5, \underline{01110},\underline{11110},10000,\underline{00001}, 01111,11111,10001}.
\end{align*}
\fi
The underlined vectors correspond to $f_0,f_1,f_2$, others are their linear combinations.
From the above set, we choose vectors with last $J_\vp=2$ zero elements.
They are $\set{\bnull^5,10000}$. 
Throwing away the last $2$ zeroes, we obtain $\hull{100}$.
\end{example}
\begin{corollary}
The GPB of CvPK can be computed as shown in Algorithm~\ref{alg:gpb}.
\end{corollary}
\ifonecol
\else
\IncMargin{1em}%
\begin{algorithm}
\caption{\texttt{GPB}($m$)}
\label{alg:gpb}

\DontPrintSemicolon
\KwIn{$m\geq 2$}
\KwOut{GPB $P^{(\vp,\mS)}$ for kernel $Q^{(n)}$, $n=2^m$, for all $\vp\in[n-2]$, $\mS\in\bS_3$}
\tcc{first loop: compute mapping $\bfT_\vp$} 
\For{$(i,j)\in [16]\times[16] $}{ \label{l:tinit0}
  $S_{0\dots 3}\gets\emptyset$\\
  \For{$(p_0^2,q_0^2) \in \mT_i \times \mT_j$}{ \label{l:pq0}
    $r_0^5 \gets p_0^2A_0 + q_0^2B_0$ \label{l:a0b0}\\
    $S_3\gets S_3 \cup r_3^5$ \\
    \lIf{$r_5=0$}{$S_2 \gets S_2 \cup \set{r_2^4}$} 
    \lIf{$r_4^5=\bnull$}{$S_1\gets S_1 \cup \set{r_1^3}$}
    \lIf{$r_3^5=\bnull$}{$S_0\gets S_0 \cup \set{r_0^2}$}  \label{l:pq1}
  }
  \For{$k\in[4]$}{
    $T_k[i][j]\gets I^{-1}(S_k)$\label{l:tinit1}
  }
}
$P\gets$ Load GPB of $Q^{(4)}$  from Table~II \label{l:gpb4}\\
\tcc{main loop: compute GPB for $Q^{(2^\lambda)}$}
\For{$\lambda=3 \dots m$}{ \label{l:main0}
  \texttt{swap}$(P,R)$ \label{l:swap}\\
  $\Lambda=2^\lambda$\\
  $P[0]\gets$\texttt{Combine}$(R[0],T_0)$ \label{l:p0}\\
  \For{$\psi = 0 \dots \Lambda/2-3$}{
    $P[2\psi+1]\gets$\texttt{Combine}$(R[\psi],T_1)$ \label{l:p1}\\
    $P[2\psi+2]\gets$\texttt{Combine}$(R[\psi],T_2)$ \label{l:p2}\\
  }
  $P[\Lambda-3]\gets$\texttt{Combine}$(R[\Lambda/2-3],T_3)$ \label{l:p3}\\
  \label{l:main1}
} 
\Return{$P[2^m-3][0..15]$}

\end{algorithm}
\DecMargin{1em}%
\fi
\begin{proof}
Let $\mT_0,\mT_1,...,\mT_{15}$ be the subspaces of $\bF^3$, indexed by operator $I:[16]\to\bS_3$, which returns $\mT_i$ by input index $i$ (for example, as given in Table~\ref{t:gpb4}).
The first loop (lines \ref{l:tinit0}--\ref{l:tinit1}) uses \eqref{eq:tdef} to compute tables $T_0,T_1,T_2,T_3:[16]\times[16]\to[16]$, which correspond to $\bfT_0$, $\bfT_{2\psi+1}$, $\bfT_{2\psi+2}$, $\bfT_{n-3}$, respectively, but work with indices $i$ instead of spaces $\mT_i$ themselves.
For example, $T_1[i][j]=l$ in the Algorithm means $\bfT_{2\psi+1}(\mT_i,\mT_j)=\mT_l$ in Theorem~\ref{t2}.

In the first loop, we run over all pairs of subspaces from $\bS_3$.
For each pair of subspaces $(\mT_i,\mT_j)$, in the internal loop (lines \ref{l:pq0}--\ref{l:pq1}) we run over all possible pairs of vectors $p_0^2$ and $q_0^2$ from these subspaces, and compute $r_0^5=pA_0+qB_0$.
In line \ref{l:a0b0} we use matrices $A_0$ and $B_0$, since $A_{2\psi+1}$, $A_{2\psi+2}$, $A_{n-3}$ are submatrices of $A_0$, the same holds for matrices $B_\vp$ (see \eqref{eq:t1vp0}--\eqref{eq:t1vp3}).
We check if the last $J_\vp$ positions of $r_0^5$ are zero.
If so, we choose the appropriate subvector of $r_0^5$, and place it in the corresponding list $S_k$.
The list $S_k$ at the end of the internal loop is equal to $\mT_l=\bfT_\vp(\mT_i,\mT_j)$.
Then, in line~\ref{l:tinit1} we perform the inverse indexing $I^{-1}$ of spaces in $\bS_3$ and obtain $l=T[i][j]$, defined above. 

In line~\ref{l:gpb4} $P$ is initialized with the GPB of kernel $Q^{(4)}$, i.e., the array $P[0..1][0..15]$ of polynomials in $x$. 
Each output value $P[\vp][i]$ is given in Table~\ref{t:gpb4} as $P^{(\vp,\mT_i)}$.

\ifonecol
\else
\IncMargin{1em}%
\begin{algorithm}
\caption{\texttt{Combine}($R, T$)}
\label{alg:combine}

\DontPrintSemicolon
\KwIn{$R[0..15]$: array of polynomials in $x$. $R[i]=P^{(\psi,\mT_i)}(x)$ for kernel $Q^{(\Lambda/2)}$\\$T[0..15][0..15]$: table with indices corresponding to specific $\bfT_\vp$}
\KwOut{$P[0..15]$: array of polynomials in $x$. $P[i]=P^{(\vp,\mT_i)}(x)$ for kernel $Q^{(n)}$}
$P[0..15]\gets 0$\\
\For{$(i,j)\in [16]\times[16] $}{
    $P[T[i][j]] \gets P[T[i][j]]+R[i]\cdot R[j]$
}
\Return{$P[0..15]$}

\end{algorithm}
\DecMargin{1em}%
\fi
In the main loop (lines~\ref{l:main0}--\ref{l:main1}) the GPB is recursively computed by Theorem~2.
At the beginning of iteration $\lambda$, array $P$ contains the GPB for kernel $Q^{({\Lambda/2})}$, where $\Lambda=2^\lambda$.
In line \ref{l:swap}, we swap $P$ and $R$ (as pointers), so after this line $R$ contains the GPB for $Q^{(\Lambda/2)}$.
Then, we compute GPB of kernel $Q^{(\Lambda)}$ and place it in array $P$.
In lines~\ref{l:p0}, \ref{l:p1}--\ref{l:p3} we use function \texttt{Combine}, defined in Alg.~\ref{alg:combine}, which applies \eqref{eq:t2} with input table $T[0..15][0..15]$ to the input GPB.
\end{proof}

Since the first loop of computing $\bfT_\vp$ in lines~\ref{l:tinit0}--\ref{l:tinit1} has constant complexity, the asymptotic complexity $C_\text{total}$ of Algorithm~\ref{alg:gpb} is  $C_\text{total}=\sum_{\lambda=3}^m C_{\text{main}}(\lambda)$, where $C_\text{main}(\lambda)$ is the complexity of the $\lambda$-th iteration
of the main loop.
The complexity $C_\text{main}(\lambda)$ is $\Lambda=2^\lambda$ times the complexity of function \texttt{Combine}.
The complexity of function \texttt{Combine} depends on current $\lambda$, because the degrees of input polynomials grow approximately as $\Theta(\Lambda)=\Theta(2^\lambda)$, and the polynomial coefficients grow as $\Theta(2^\Lambda)$.
Function \texttt{Combine} consists in $256$ multiplications of such polynomials.
Assume that we multiply these polynomials and their integer coefficients straightforwardly.
Then, polynomial multiplication includes $\Theta(\Lambda^2)$ multiplications of integers. Each integer has length $\Theta(\Lambda)$
and their straightforward multiplication has complexity $\Theta(\Lambda^2)$.
Thus, the complexity of \texttt{Combine} function is asymptotically
$
C_\text{combine}(\lambda)\approx\Lambda^4=16^{\lambda}.
$
The total complexity is
$$
C_\text{total}=\sum_{\lambda=3}^{m}C_\text{main}(\lambda)=\sum_{\lambda=3}^{m}\Theta(2^\lambda\cdot16^\lambda)=\Theta(32^m)=\Theta(n^5).
$$
One can reduce this complexity to $\Theta(n^3\log^2n)$ by using fast algorithms for multiplication of big integers and polynomials.

\subsection{Converting GPB to PB \label{s:gpb2pb}}
Polarization behaviour $P^{(\vp)}(x)$ (see Definition~\ref{d:pb}) is the weight spectrum of all erasure configurations $\mG$ that erase $u_\vp$.
This means that linear combination $(1,0,0)\bigcdot u_\vp^{\vp+2}$ must not be recoverable, so $(1,0,0)\notin \chi_\vp(\mG)$. 

More formally, let $\Xi$ be the set of all erasure configurations $\mG$ such that
$(1,\bnull^{n-\vp-1})\notin\cs K_{\overline{[\vp]},\overline{\mG}}$.
Then,
$$
P^{(\vp)}(x)=\sum_{\mG\in\Xi}x^{|\mG|}.
$$
Observe that $\mG\in\Xi\iff(1,\bnull^{n-\vp-1})\notin\cs K_{\overline{[\vp]},\overline{\mG}}\implies (1,\bnull^2)\notin\chi_\vp(\mG)$.
The reverse implication also holds and \mbox{$\mG\in\Xi\iff(1,\bnull^2)\notin\chi_\vp(\mG)$}, which leads
to
\begin{align*}
\Xi=\bigcup_{\mS\in\bS_3:(1,0,0)\notin\mS}\chi^{-1}_\vp(\mS).
\end{align*}
The last two equations imply
\begin{align}
P^{(\vp)}(x)=\sum_{\mS\in\bS_3:(1,0,0)\notin\mS}P^{(\vp, \mS)}(x),
\label{eq:pb2gpb}
\end{align}
where $P^{(\vp,\mS)}(x)$ is the GPB of $K$.
Formula \eqref{eq:pb2gpb} is defined for $\vp\leq n-3$.
Polynomials $P^{(n-2)}(x)$ and $P^{(n-1)}(x)$ can be obtained by
\begin{align}
P^{(n-2)}(x)&=\sum_{\mS\in\bS_3:\forall a\in\bF: (a,1,0)\notin\mS}P^{(\vp, \mS)}(x) \label{eq:pb2gpb1}\\
P^{(n-1)}(x)&=\sum_{\mS\in\bS_3:\forall a\in\bF^2: (a,1)\notin\mS}P^{(\vp, \mS)}(x) \label{eq:pb2gpb2}
\end{align}

Computing each of \eqref{eq:pb2gpb}--\eqref{eq:pb2gpb2} consists of adding respectively $11$, $8$ and $5$ polynomials of degree $n$ with
integer coefficients of length $O(n)$, so the total complexity of converting GPB to PB is $n\cdot O(n^2)=O(n^3)$, which does
not affect the total asymptotic complexity.

\subsection{Polarization Rate of CvPK \label{s:secvpk}}
Polarization rate of an $n\times n$ polarizing kernel $K$ can be obtained as \cite{korada2010polar}
\begin{align}
E(K)=\frac{1}{n}\sum_{i=0}^{n-1}\log_nd_i,
\label{eq:rho}
\end{align}
where $d_i$ is called the $i$-th partial distance
and is defined by
\begin{align}
d_i=\min_{u_{i+1}^{n-1}\in\bF^{n-i-1}}\wt\left((1,u_{i+1}^{n-1})K_{\overline{[i]},*}\right).
\label{eq:di}
\end{align}
Observe that $d_i$ is the minimum degree of a non-zero monomial in $P^{(i)}(x)$ (see e.g. \cite{morozov2019distance} for the proof).
So, the values of $d_i$ for $Q^{(n)}$ can be easily obtained from PB of $Q^{(n)}$.

\subsection{Row-permuted CvPKs \label{s:perm}}

We observed that one can permute rows of $Q^{(n)}$ and obtain better scaling exponent.
Moreover, we found a permutation that does not affect much neither the kernel processing algorithm, nor the Alg.~\ref{alg:gpb} of computing GPB of a CvPK.
We start with a proposition, which shows how to construct kernel $\tK$ from a given $K$ with improved polarization rate in general.
\begin{proposition}
\label{p:swapd}
Consider $n\times n$ kernel $K$ and $i\in[n]$, for which $d_i\geq d_{i+1}$.
Swap rows $i$ and $i+1$ and denote the resulting kernel by $\tK$.
Then, $E(\tK)\geq E(K)$.
\end{proposition}
\begin{proof}
Denote disjoint sets
\begin{align*}
\mA=\set{c_0^{n-1}\!=\!(1,0,u_{i+2}^{n-1})K_{\overline{[i]},*}\;\big|\;u_{i+2}^{n-1}\in\bF^{n-i-2}}\\
\mB=\set{c_0^{n-1}\!=\!(0,1,u_{i+2}^{n-1})K_{\overline{[i]},*}\;\big|\;u_{i+2}^{n-1}\in\bF^{n-i-2}}\\
\mC=\set{c_0^{n-1}\!=\!(1,1,u_{i+2}^{n-1})K_{\overline{[i]},*}\;\big|\;u_{i+2}^{n-1}\in\bF^{n-i-2}}
\end{align*}

For set of vectors $S$, denote by $\uwt(S)$ the minimum weight of vector from $S$.
Observe that 
\begin{align*}
d_i&=\uwt(\mA\cup\mC)=\min\set{\uwt(\mA), \uwt(\mC)}\\
d_{i+1}&=\uwt(\mB)\leq d_i \implies \uwt(\mB)\leq\uwt(\mC)\\
\td_i&=\uwt(\mB\cup\mC)=\uwt(\mB)=d_{i+1}\\
\td_{i+1}&=\uwt(\mA)\geq\uwt(\mA\cup\mC)=d_i,
\end{align*}
where $\td_0^{n-1}$ are the partial distances of $\tK$.
Thus, $\td_i=d_{i+1}$, $\td_{i+1}\geq d_{i}$.
Obviously, $\td_j=d_j$ for $j\notin\set{i,i+1}$.
Recalling \eqref{eq:rho}, one obtains $E(\tK)\geq E(K)$.
\end{proof}
We can apply the proposition multiple times and obtain bubble sorting of rows by their partial distances.
\begin{corollary}
\label{c:ok}
Denote by $\overline K$ kernel with rows of $K$, sorted by $d_i$ in ascending order. Then, $E(\overline K)\geq E(K)$.
\end{corollary}
\begin{corollary}
\label{c:swapp}
Let $d_i=d_{i+1}=w$ and $P^{(i)}_w<P^{(i+1)}_w$, where $d_*$ and $P^{(*)}$ are partial distances and PB of kernel $K$, respectively.
Swap rows $i$ and $i+1$ and denote the resulting kernel by $\tK$.
Then, $\tP^{(i+1)}_w\leq P^{(i)}_w<P_w^{(i+1)}\leq \tP^{(i)}_w$, where $\tP^{(*)}$ is the PB of $\tK$.
\end{corollary}
\begin{proof}
For set of vectors $S$, denote by $S_w$ the set of all vectors from $S$ with weight $w$.
Then, $P^{(i)}_w=|\mA_w|+|\mC_w|$, $P^{(i+1)}_w=|\mB_w|$ and
\begin{align*}
\tP^{(i)}_w&=|\mB_w|+|\mC_w|\geq P_w^{(i+1)}\\
\tP^{(i+1)}_w&=|\mA_w|\leq P_w^{(i)}.
\end{align*}
Thus, $\tP^{(i+1)}_w\leq P^{(i)}_w<P_w^{(i+1)}\leq \tP^{(i)}_w$.
\end{proof}
\begin{remark}
\label{r:betterse}
Intuitively, in the pair of subchannels $W^{(i)}$ and $W^{(i+1)}$, induced by the kernel from Corollary~\ref{c:swapp}, the ``bad'' one becomes ``worse'' and the ``good''
one becomes ``better'' by swapping the rows.
Intuition suggests that this leads to $\mu(\tK)\leq\mu(K)$.
Also, by Proposition~\ref{p:swapd}, $E(\tK)\geq E(K)$.
\end{remark}
\begin{remark}
\label{r:tq}
We observed that for CvPK $d_{2i}\geq d_{2i+1}$ for $i=2..n/2-3$.
Denote by $\tQ^{(n)}$ kernel $Q^{(n)}$ with swapped $2i$-th and $(2i+1)$-th rows for $i=2..n/2-3$.
One can easily obtain PB $\tP^{(\vp)}(x)$ of kernel $\tQ^{(n)}$  from GPB $P^{(\vp,\mS)}(x)$ of kernel $Q^{(n)}$ by similar to \eqref{eq:pb2gpb}--\eqref{eq:pb2gpb2} formulae:
\begin{align}
\tP^{(\vp)}(x)&=P^{(\vp)}(x), \text{ for }\vp\leq 3 \text{ or } \vp\geq n-4, \label{eq:tp0}\\
\tP^{(2i)}(x)&=\sum_{\mS\in\bS_3: \forall a\in\bF: (0,1,0)\notin\mS} P^{(2i, \mS)}(x), \label{eq:tp1}\\
\tP^{(2i+1)}(x)&=\sum_{\mS\in\bS_3: \forall a\in\bF: (1,a,0)\notin\mS}P^{(2i, \mS)}(x). \label{eq:tp2}
\end{align}
\end{remark}

Also, SC decoding for $\tQ^{(n)}$ is very similar to SC decoding for $Q^{(n)}$, as described in Appendix~\ref{s:tq}.
\section{Numerical Results \label{s:nr}}
\subsection{Scaling Exponent and Polarization Rate}
\ifonecol
\else
\begin{table}
\centering 
\caption{Polarization rate $E$ and scaling exponent $\mu$ of convolutional polar kernels of size $n$. Best $\mu$ corresponds to a known kernel with the lowest scaling exponent
from \cite{yao2019explicit}.}
\label{t:se}               
\begin{tabular}{|c|c|c|c|c|c|c|}
\hline
$n$     & $\!E(Q^{(n)})\!$ & $\!\underline E(B^{(n)})\!$ & $\!\mu(Q^{(n)})\!$ & $\!\mu(\tQ^{(n)})\!$ & $\!\mu(\overline Q^{(n)})\!$ & best $\!\mu\!$   \\ \hline
4       & 0.5              &  0.5                         & 3.627              & 3.627                & 3.627                        & 3.627           \\ \hline
8       & 0.5              &  0.5                         & 3.577              & 3.577                & 3.577                        & 3.577           \\ \hline
16      & 0.50914          &  0.51828                     & 3.470              & 3.409                & 3.400                        & 3.346           \\ \hline
32      & 0.52194          &  0.53656                     & 3.382              & 3.316                & 3.153                        & 3.122           \\ \hline
64      & 0.52923          &  0.56427                     & 3.333              & 3.283                &                              & 2.87            \\ \hline
128     & 0.53482          &  0.58775                     & 3.310              & \textbf{3.277}       &                              &                 \\ \hline
256     & 0.53865          &  0.61333                     & \textbf{3.303}     & \textbf{3.283}       &                              &                 \\ \hline
512     & 0.54106          &  0.63559                     & \textbf{3.308}     & \textbf{3.296}       &                              &                 \\ \hline
1024    & 0.54260          &  0.65688                     & \textbf{3.317}     & \textbf{3.311}       &                              &                 \\ \hline
\end{tabular}              
\end{table}               
\begin{table}
\centering
\caption{Polarization rate $E$ of large CvPK of size $n$.}
\label{t:pr}
\begin{tabular}{|c|c|c|}
\hline
$n$     & $E(Q^{(n)})$       & $\!\underline E(B^{(n)})\!$   \\ \hline        
2048    & 0.54351            & 0.67558                       \\ \hline
4096    & 0.54398            & 0.69274                       \\ \hline  
8192    & \textbf{0.54414}   & 0.70802                       \\ \hline  
16384   & \textbf{0.54408}   & 0.72187                       \\ \hline  
32768   & \textbf{0.54386}   & 0.73432                       \\ \hline  
65536   & \textbf{0.54353}   & 0.74564                       \\ \hline  
\end{tabular}
\end{table}
\fi

In Table~\ref{t:se} one can see the computed values of scaling exponent for BEC and polarization rate for kernels $Q^{(n)}$ and $\tQ^{(n)}$.
Since PB for these kernels can be obtained by polynomial algorithm, we obtain scaling exponent for these kernels for $n$ up to $1024$.
Remark~\ref{r:betterse} suggests $\mu(\tQ^{(n)})\leq\mu(Q^{(n)})$.
Although we do not prove this inequality, one can see in Table~\ref{t:se} that it indeed holds for all $n\leq 1024$, becoming strict for $n\geq 16$.

We also provide scaling exponent for kernel $\overline Q^{(n)}$, consisting of rows of $Q^{(n)}$, sorted by partial distances, as described in Corollary~\ref{c:ok}.
Also some adjacent rows were sorted by $P^{(i)}_w$ as described in Corollary~\ref{c:swapp}.
The specific row permutations $\pi_{16}$ and $\pi_{32}$, corresponding to $\overline Q^{(16)}_{i,*}=Q^{(16)}_{\pi_{16}(i),*}$ and $\overline Q^{(32)}_{i,*}=Q^{(32)}_{\pi_{32}(i),*}$, are
\small
\ifonecol
\else
\begin{align*}
\pi_{16}=(&0, 1, 2, 3, 5, 4, 7, 6, 10,8, 11,9, 12,13,14,15),\\
\pi_{32}=(&0,1,2,3,6,4,9,7,13,5,20,8,14,11,18,15,\\&16,10,23,19,24,12,26,17,25,21,27,22,28,29,30,31).
\end{align*}
\fi
\normalsize
One can see that, unlike the case of kernel $\tQ^{(n)}$, the rows order in $\overline Q^{(n)}$ is very different from the original order in $Q^{(n)}$.
We found no formulae  to obtain PB of $\overline Q^{(n)}$ from the GPB of $Q^{(n)}$, similar to \eqref{eq:tp0}--\eqref{eq:tp2}.
We obtain PB of $\overline Q^{(n)}$ for $n\leq 32$ by brute force.
One can see that the proposed row permutation leads to smaller scaling exponent, comparable to the best known \cite{yao2019explicit}.
For all studied cases, $E(Q^{(n)})=E(\tQ^{(n)})=E(\overline Q^{(n)})$.
In Table~\ref{t:pr} one can see polarization rate of large CvPKs, obtained by a simplified procedure \cite{morozov2019simplified}.
We also provide a lower bound $\underline E(B^{(n)})$ of polarization rate of BCH kernels $B^{(n)}$, where partial distances are lower-bounded by constructive distances of extended BCH codes, generated by the bottom rows of $B^{(n)}$.

What is counter-intuitive is that $\mu(Q^{(512)})>\mu(Q^{(256)})$, $\mu(Q^{(256)})>\mu(Q^{(128)})$ and $E(Q^{(16384)})<E(Q^{(8192)})$.
Intuitively, for the kernels which have the same structure, the larger is the kernel, the better polarization properties it has.
Although results for scaling exponent may be imprecise due to numerical errors, computing polarization rate is simple and numerically stable.
On the other hand, if the scaling exponent of $Q^{(n)}$ tended to $2$ with $n\to\infty$, that would mean existence of codes of lengths $N=n^M$, which achieve optimal scaling exponent with decoding complexity $O(N\log N)$. 
This sounds too good to be true.

Polarization rate in \cite{ferris2013branching} was heuristically estimated to be around 0.62, although no rigorous proof of channel polarization was provided.
In our scenario, channel polarization follows from the general proof for the case of large kernels, obtained in \cite{korada2010polar}, and we obtain a precise estimate of the polarization rate.

\subsection{Performance of Polar Codes with CvPK}
\ifonecol
\else
\begin{figure}
\includegraphics[width=0.5\textwidth]{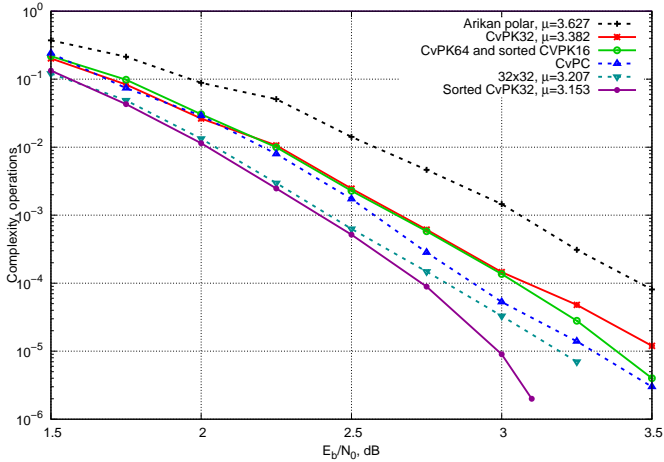}
\caption{Performance of $(1024,512)$ polar codes with various CvPKs (solid) and other kernels (dashed) under SC decoding.}
\label{fig:cvpk1k}
\end{figure}
\begin{figure}
\includegraphics[width=0.5\textwidth]{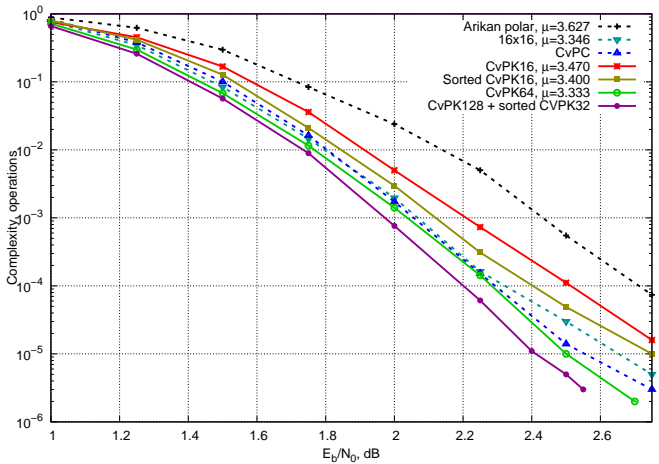}
\caption{Performance of $(4096,2048)$ polar codes with various CvPKs (solid) and other kernels (dashed) under SC decoding.}
\label{fig:cvpk4k}
\end{figure} 
\fi
Fig.~\ref{fig:cvpk1k} presents the SC decoding performance of $(1024,512)$ codes, corresponding to polarizing transformations $F^{\otimes 10}$, $Q^{(32)\otimes 2}$, $Q^{(64)} \otimes \overline Q^{(16)}$, $Q^{(1024)}$, $K_{3}^{\otimes 2}$, $\overline Q^{(32)\otimes 2}$, the order is the same as in the legend.
Kernel $K_{3}$ is from \cite{trofimiuk2018efficient}, $\mu(K_{3})=3.207$ and $E(K_{3})=0.52925$.
The design SNR is $E_b/N_0=2.75$ dB.
One can see that polar code with sorted $32\times 32$ CvPK $\overline Q^{(32)}$ outperforms polar codes with other kernels and the CvPC due to its lower scaling exponent, even though it does not have the highest polarization rate.
The polarizing transformation $Q^{(64)} \otimes \overline Q^{(16)}$ corresponds to a polar code with mixed kernels.
The definition of polar codes with mixed kernels can be obtained by replacing $K^{\otimes M}$ in \eqref{eq:pcdef} with $K_1\otimes ...\otimes K_M$.

Fig.~\ref{fig:cvpk4k} presents the SC decoding performance of $(4096,2048)$ codes with polarizing transformations $F^{\otimes 12}$, $K_{2}^{\otimes 3}$, $Q^{(4096)}$ (dashed),
and $Q^{(16)\otimes 3}$, $\overline Q^{(16)\otimes 3}$, $Q^{(64)\otimes 2}$, $Q^{(128)}\otimes \overline Q^{(32)}$ (solid).
Kernel $K_{2}$ is from \cite{trofimiuk2019reduced}, $\mu(K_{2})=3.346$ and $E(K_{2})=0.51828$.
The design SNR is $E_b/N_0=2.25$~dB.
One can see that polar code with $Q^{(128)}\otimes \overline Q^{(32)}$ has the best performance.

Polar codes with Arikan kernel were constructed using Gaussian approximation \cite{trifonov2012efficient}, other codes were constructed using Monte-Karlo simulations.
For kernels $K_2$ and $K_3$ kernel processing is defined in \cite{trofimiuk2018efficient, trofimiuk2019reduced}.
Efficient processing of $\overline Q^{(n)}$ is done by the general trellis-based algorithm \cite{trifonov2019trellis}.

For $Q^{(n)}$ the kernel processor is  the SC decoder from \cite{morozov2018efficient}.
Note that for CvPK $Q^{(n)}$ the complexity of kernel processing is $O(n\log n)$, in contrast with an arbitrary kernel of size $n$, where, in general, the complexity is $O(2^n)$.
Observe also that processing of kernel $\widetilde Q^{(n)}$ can be also done by the SC decoder for CvPC  with swapping adjacent phases on layer $m$.

The complexity of SC decoding for $(1024,512)$ codes from Fig.~\ref{fig:cvpk1k} is presented in Table~\ref{t:compl}, together with the SC decoding frame error probability
(FER) at $E_b/N_0=3$ dB.
Note that the decoding complexity increases monotonously with the decrease of error probability.
This approves the fact that CvPKs are competitive compared to other polarization kernels. 
Regarding distance properties of the obtained $(1024,512)$ polar codes, all codes have the same minimum distance of $16$,
so we also present the error coefficient, i.e., the number of codewords of weight $16$.
One can see non-monotonous dependence of FER on the error coefficient, since SC decoding is not near-ML decoding.
\begin{table}
\footnotesize
\caption{ SC decoding complexity of $(1024,512)$ polar codes, and an approximate number of minimum-weight codewords, found by \cite{canteaut98new}. In all cases $d=16$.}
\label{t:compl}
\centering
\begin{tabular}[width=\textwidth]{|c|c|c|c|c|}
\hline
Polar. transform & Compl. & FER at 3 dB & Err. coeff. &Decoder 
\\ \hline
 ${\begin{pmatrix}10\\11\end{pmatrix}}^{\otimes 10}$  & $1.4\cdot 10^4$ & $1.6\cdot 10^{-3}$ &49344&  \cite{arikan2009channel}
\\ \hline
$Q^{(32)}\otimes Q^{(32)}$ & $6.6\cdot 10^4$& $1.5\cdot 10^{-4}$ & 19648 & \cite{morozov2020efficient}
\\ \hline
 $Q^{(64)}\otimes \overline Q^{(16)}$ & $8.4\cdot 10^4$ &$1.4\cdot 10^{-4}$ & 18624 & \cite{morozov2020efficient,trifonov2019trellis}
\\ \hline
$Q^{(1024)}$  & $2.4\cdot 10^5$ & $5.3\cdot 10^{-5}$ & 2240 & \cite{morozov2020efficient}
\\ \hline
$K_3\otimes K_3$ & $4.4\cdot 10^5$ & $3.3\cdot 10^{-5}$ & 1984 & \cite{trofimiuk2018efficient}
\\ \hline
 $\overline Q^{(32)}\otimes \overline Q^{(32)}$ &$1.1\cdot 10^6$ & $9.0\cdot10^{-6}$ & 4288 &  \cite{trifonov2019trellis}
\\\hline
\end{tabular}
\end{table}
\section{Conclusions}
In this paper, a family of convolutional polar kernels (CvPKs) of size $n=2^m$ was proposed together with the polynomial-complexity algorithm for computing polarization behaviour, scaling exponent and polarization rate. The kernels are based on convolutional polar codes. The proposed algorithm enables one to study
polarization properties of CvPKs of size up to $1024\times 1024$.
Polarization properties of convolutional polar kernels are getting worse, starting from sufficiently large size. 
The row permutation operation was suggested, that can improve scaling exponent of CvPK.
The proposed family of kernels allow kernel processing with complexity $O(n \log n)$ as the kernel size  $n$ tends to infinity.

\appendices
\section{Proof of Theorem 1 \label{a:t1}}
Let us prove the theorem for the case of $\vp=2\psi+1$, corresponding to \eqref{eq:t1vp1}.
If the receiver knows $u_0^{2\psi}$, then it knows $x_0^{\psi-1}$ and $z_0^{\psi-1}$ by \eqref{eq:xz}.
Denote the stripped kernels without rows, corresponding to known (already estimated) symbols, and without columns, corresponding to erased symbols, by
$\hat Q=Q^{(n)}_{\overline{[\vp]},\overline{\mE}}, \;\hat Q'=Q^{(n/2)}_{\overline{[\psi]},\overline{\mE'}}, \;\hat Q''=Q^{(n/2)}_{\overline{[\psi]},\overline{\mE''}}$.
Denote $k=n-\vp$, $k'=\frac{n}{2}-\psi$.
and $\overline w=n-|\mE|$, $\overline w'=n/2-|\mE'|$, $\overline w''=n/2-|\mE''|$.
Then, the size of $\hat Q$ is $k\times\overline w$, the sizes of $\hat Q'$ and $\hat Q''$ are $k'\times \overline w'$ and $k'\times \overline w''$.

Denote the transition matrices $X^{(n)}$ and $Z^{(n)}$ without rows and columns, corresponding to known symbols, by
$
\hat X=X^{(n)}_{\overline{[\vp]},\overline{[\psi]}},\; \hat Z=Z^{(n)}_{\overline{[\vp]},\overline{[\psi]}}.
$
The sizes of $\hat X$ and $\hat Z$ are $k\times k'$.
Using above notations,  one obtains 
$
\hat Q=(\hat X\hat Q',\hat Z \hat Q'').
$

The theorem for the case of \eqref{eq:t1vp1}  now can be reformulated as
$(p_0^2,\bnull^{k-3})\in\cs\hat Q$, if and only if there exists $(p',\bnull^{k'-3})\in\cs\hat Q'$, $(p'',\bnull^{k'-3})\in\cs\hat Q''$, such that $(p,\bnull^{2})=p'A+p''B$.
Observe that $(p,\bnull^{k-3})\in\cs\hat Q$ iff there exists $q$, s.t.
\begin{align}
(p,\bnull^{k-3})=\hat Qq^T\!=\!(\hat X \hat Q', \hat Z \hat Q'')q^T=\hat X \hat Q'q'^T+ \hat Z \hat Q''q''^T,
\label{eq:vq}
\end{align}
where $q=(q',q'')$.
Denote $a=\hat Q'q'^T$, $b=\hat Q''q''^T$.
Note that $a\in\cs\hat Q'$ and $b\in\cs\hat Q''$.
Thus, such $q$ in \eqref{eq:vq} exists iff 
\begin{align}
\exists a\in\cs\hat Q', b\in\cs\hat Q'': 
(p,\bnull^{k-3})^T=\hat X a^T+ \hat Z b^T.
\label{eq:t1ab}
\end{align}
The r.h.s. of  \eqref{eq:t1ab} are $a_i+b_i$ for the $2i$-th equation, and $a_i+a_{i+1}+b_i$ for the $(2i+1)$-th equation.
The first five equations of \eqref{eq:t1ab} are
{\allowdisplaybreaks
\ifonecol
\else
\begin{align}
a_0+b_0=p_0,\;
a_0+a_1+b_0&=p_1,\;
a_1+b_1=p_2,\nonumber\\
a_1+a_2+b_1&=0, \;
a_2+b_2=0. \label{eq:t1p0}
\end{align}
\fi
Then, there are $k-5$ equations of the form
\begin{align*}
a_2+a_3+b_2=0 & \iff a_3=0 \text{ (since } a_2+b_2=0\text{)}\\
a_3+b_3=0 & \iff b_3=0 \text{ (since } a_3=0\text{)}\\
a_3+a_4+b_3=0 & \iff a_4=0 \text{ (since } a_3+b_3=0\text{)}
\end{align*}
}and so on. Thus, 
$
a_3^{k'-1}=b_3^{k'-1}=\bnull.
$
Since $a\in\cs\hat Q', b\in\cs\hat Q''$, by Def.~\ref{d:chi} the last $k-5$ equations are equivalent to $a_0^2\in\chi_\psi(\mE')$,
$b_0^2\in\chi_\psi(\mE')$ for kernel $Q^{(n/2)}$.
Combining this with \eqref{eq:t1p0}, one can prove the theorem, since \eqref{eq:t1} with \eqref{eq:t1vp1} are precisely \eqref{eq:t1p0}, written in matrix form for $p'=a$ and $p''=b$.

The other cases of $\vp$ can be proved similarly.

\section{Proof of Theorem 2}
\label{a:t2}
First, fix some $\mE\in[n]$.
Let $\chi_\psi(\mE')=\mS'$ and $\chi_\psi(\mE'')=\mS''$.
By Theorem~\ref{t1}, $p_0^2\in\chi_\vp(\mE)\iff \exists p'\in\chi_\psi(\mE'), p''\in\chi_\psi(\mE'')$, such that $(p_0^2,\bnull^{J_\vp})\!=\!p'A_\vp\!+\!p''B_\vp$.
Substituting $\mS'=\chi_\psi(\mE')$ and $\mS''=\chi_\psi(\mE'')$ one obtains precisely the conditional part of \eqref{eq:tdef}.
Thus,
$
\chi_\vp(\mE)=\bfT_\vp(\chi_\psi(\mE'),\chi_\psi(\mE'')).
$
Using Definition~\ref{d:chi}, rewrite \eqref{eq:gpbdef} as
\begin{align}
P^{(\vp, \mS)}(x)=\sum_{\mE\in\chi^{-1}_\vp(\mS)}x^{|\mE|}.
\label{eq:gpbsum}
\end{align}
Observe that $\mE\in\chi^{-1}_\vp(\mS)$ iff $\exists (\mS', \mS'')\in\bfT_\vp^{-1}(\mS)$, such that $\mE'\in\chi^{-1}_\psi(\mS')$ and  $\mE''\in\chi^{-1}_\psi(\mS'')$.
The erasure configuration $\mE$ is bijectively defined by its ``halves'' $\mE'$ and $\mE''$, so can replace summation over $\chi^{-1}_\vp(\mS)$ in \eqref{eq:gpbsum} by two independent summations over $\chi^{-1}_\psi(\mS')$ and $\chi^{-1}_\psi(\mS'')$.
Obviously, $|\mE|\!=\!|\mE'|\!+\!|\mE''|$.
Thus, 
\ifonecol
\else
\begin{align*}
&P^{(\vp, \mS)}(x)=
\!\!\sum_{(\mS',\mS'')\in\bfT^{-1}_\vp(\mS)} \sum_{\mE'\in\chi^{-1}_\psi(\mS')}\sum_{\mE''\in\chi^{-1}_\psi(\mS'')}x^{|\mE'|+|\mE''|}\\
&=\!\!\sum_{(\mS',\mS'')\in\bfT^{-1}_\vp(\mS)}\!\!\left(\sum_{\mE'\in\chi^{-1}_\psi(\mS')}x^{|\mE'|}\right)\!\cdot\!\left(\sum_{\mE''\in\chi^{-1}_\psi(\mS'')}x^{|\mE''|}\right)\\
&=\!\!\sum_{(\mS',\mS'')\in\bfT^{-1}_\vp(\mS)}R^{(\psi,\mS')}(x)\cdot R^{(\psi,\mS'')}(x).
\end{align*}
\fi

\section{On Decoding of CvPC with Matrix  $\tQ^{(n)}$}
\label{s:tq}
The decoder for convolutional polar codes with matrix $Q^{(n)}$ (e.g. \cite{morozov2018efficient}) computes at each phase
$\vp$ the vector log-likelihood 
\begin{align}
L_\vp[a,b,c]=\ln\!\!\!\!\max_{u_{\vp+3}^{n-1}\in\bF^{n-\vp-3}}\!W^n\left((\hat u_0^{\vp-1},a,b,c,u_{\vp+3}^{n-1})Q^{(n)}|y\right).
\label{eq:lvp}
\end{align}
The output LLR for symbol $u_\vp$, needed for hard decision, is defined as 
\begin{align}
S_\vp=\ln\frac{\max_{u_{\vp+1}^{n-1}} W^n\left((\hat u_0^{\vp-1},0,u_{\vp+1}^{n-1})Q^{(n)}|y\right)}
{\max_{u_{\vp+1}^{n-1}} W^n\left((\hat u_0^{\vp-1},1,u_{\vp+1}^{n-1})Q^{(n)}|y\right)},
\label{eq:svp}
\end{align}
and can be computed by marginalization
\begin{align*}
S_\vp=\max_{b,c}L_\vp[0,b,c]-\max_{b,c}L_\vp[1,b,c].
\label{eq:marg}
\end{align*}

Matrix $\tQ^{(n)}$ is obtained from matrix $Q^{(n)}$ by swapping some of pairs of adjacent rows $(2i,2i+1)$.
Formally,
\begin{align*}
\tQ_{2i,j}=\begin{cases}
Q_{2i,j} &i\notin\mJ\\
Q_{2i+1,j} &i\in\mJ
\end{cases},
\;\;\;\;
\tQ_{2i+1,j}=\begin{cases}
Q_{2i+1,j} &i\notin\mJ\\
Q_{2i,j} &i\in\mJ
\end{cases}
\end{align*}
where we denote by $\mJ\subset[n/2]$ the set of all $i$, for which rows $2i$ and $2i+1$ are swapped in $\tQ^{(n)}$.
In \eqref{eq:svp}, replace $Q^{(n)}$ with $\tQ^{(n)}$ and denote corresponding LLR by $\tS_\vp$.
Then, $S_{2i}=\tS_{2i}$ and $S_{2i+1}=\tS_{2i+1}$ for $i\notin\mJ$.

For $i\in\mJ$, values of $\tS_{2i}$ and $\tS_{2i+1}$ can be also obtained from vector log-likelihoods \eqref{eq:lvp} with
the only change in marginalization:
\begin{align*}
\tS_{2i}&=\max_{a,c}L_{2i}[a,0,c]-\max_{a,c}L_{2i}[a,1,c]\\
\tS_{2i+1}&=\max_{c}L_{2i}[0,\hat u_{2i},c]-\max_{c}L_{2i}[1,\hat u_{2i},c]
\end{align*}

So, the only difference between decoding with $Q^{(n)}$ and decoding with $\tQ^{(n)}$ is in final marginalization when converting
vector log-likelihood to the output LLR.

\bibliographystyle{IEEEtran}

\end{document}